\documentclass[3p,10pt,final]{elsarticle0}
\usepackage[notref,notcite]{showkeys} % Disabled by option final

\usepackage{pifont}
\usepackage{amsfonts}
\usepackage{amsmath}
\usepackage{amssymb}
\usepackage{amsthm}
\usepackage{latexsym}
\usepackage{graphicx}
\usepackage{epsf}
\usepackage{latexsym}

\usepackage[latin1]{inputenc}

  
\newtheorem{remark}{Remark}[section]

\newtheorem{proposition}{Proposition}[section]

\newcommand{\DSS}{\displaystyle}

\newtheorem{thm}{Theorem}[section]
\newtheorem{cor}{Corollary}[section]

\newcommand{\half}{\frac{1}{2}}
\newcommand{\Cset}{\mathbb{C}}

\newcommand{\Rset}{\mathbb{R}}

\DeclareMathOperator{\tr}{tr}

\newcommand{\ie}{\hbox{\it i.e.\/}, }

\begin{document}
\begin{frontmatter}
\title{Global passive system approximation}
\author{L. Knockaert\fnref{label01}}
\ead{luc.knockaert@intec.ugent.be}
\address{Dept. Information Technology, IBCN, Ghent University
\\
Gaston Crommenlaan 8, PB 201, B-9050 Gent, Belgium}
\fntext[label01]{Corresponding author : tel. +3292643328, fax
+3292649969. This work was supported by a grant of the Research
Foundation-Flanders (FWO-Vlaanderen)}
\pagestyle{plain} \setcounter{page}{1}
\begin{abstract}
In this paper we present a new approach towards global passive
approximation in order to find a passive transfer function $G(s)$
that is nearest in some well-defined matrix norm sense to a
non-passive transfer function $H(s).$ It is based on existing
solutions to pertinent matrix nearness problems. It is shown that
the key point in constructing the nearest passive transfer
function, is to find a good rational approximation of the
well-known ramp function over an interval defined by the minimum
and maximum dissipation of $H(s).$ The proposed algorithms rely on
the stable anti-stable projection of a given transfer function.
Pertinent examples are given to show the scope and accuracy of the
proposed algorithms.
\end{abstract}
\begin{keyword}
%% keywords here, in the form:
Passivity \sep positive-real lemma  \sep rational approximation
%% PACS codes here, in the form:
%%\PACS 84.30.-r \sep 84.32.-y
%% MSC codes here, in the form:
%%\MSC[2010] 93D09 \sep 41A20
%%or \MSC[2008] code \sep code (2000 is the default)
\end{keyword}
\end{frontmatter}
\section{INTRODUCTION}
For linear time-invariant systems, passivity guarantees stability
and the possibility of synthesis of a transfer function by means
of a lossy physical network of resistors, capacitors, inductors
and transformers \cite{and}. Therefore, passivity enforcement
\cite{sgt1} and passification (passivation) \cite{passif} have
become important issues in recent years
\cite{damar,dumi,coelho,sacna,tanji}, especially as more and more
software tools render transfer functions which need passivity
enforcement as a postprocessing step in order to generate reliable
physical models. However, most of the techniques
\cite{sgt1,passif,damar,dumi,coelho,sacna} are local perturbative
and/or feedback approaches with fixed poles, while \cite{tanji} is
based on Fourier approximation, yielding passivated systems with a
large number of poles. \newline In this paper we present a new
global approach in the sense that we find a passive transfer
function $G(s)$ that is nearest in a well-defined matrix norm
sense to a non-passive transfer function $H(s).$ It is based on
existing solutions to some pertinent matrix nearness problems
\cite{Higham89matrixnearness,halmos}. We show that the key point
in constructing the nearest passive transfer function $G(s) ,$ is
to find a good rational approximation for the ramp function
$\max(0,x)$ over an interval defined by the minimum and maximum
dissipation of the non-passive transfer function $H(s).$ It is
also shown that in the Chebyshev or minimax sense this requires
finding a rational Chebyshev approximation of the square root
$\sqrt{x}$ over the interval $[0,1].$ The proposed algorithms rely
heavily on the stable anti-stable projection \cite{kag,Sajoveli}
of a given transfer function. Finally, five pertinent examples,
both SISO and MIMO, are given to show the accuracy and relevance
of the proposed algorithms.
\section{PASSIVITY AND DISSIPATION}
\label{sec:main} Notation : Throughout the paper $X^T$ and $X^H$
respectively denote the transpose and Hermitian transpose of a
matrix $X ,$ and $I_n$ denotes the identity matrix of dimension
$n.$ The Frobenius norm is defined as $ \|X\|_F = \sqrt{ \tr X^H X
}$ and the spectral norm (or 2-norm or maximum singular value) is
defined as $\|X\|_2 = \sqrt{ \lambda_{\mathrm{max}} ( X^H X) } .$
It is easy to show that $\|X^H\|_F = \|X\|_F$ and  $\|X^H\|_2 =
\|X\|_2 .$ For two Hermitian matrices $X$ and $Y,$ the matrix
inequalities $ X
> Y$ or $X \geq Y$ mean that $X-Y$ is respectively positive definite
or positive semidefinite. The closed right halfplane $ \Re e \,
[s]
\geq 0 $ is denoted $\Cset_+ $. \\
For the real system with minimal realization
\begin{subequations}
\label{sub1}
\begin{eqnarray}
\dot{x} &=& A x + B u  \\
y &=& C x + D u
\end{eqnarray}
\end{subequations}
where $B \neq 0,$ $C \neq 0 $ are respectively $n \times p $ and
$p \times n $ real matrices and $A \neq 0$ is a $n \times n $ real
matrix, to be passive, it is required that the $p \times p $
transfer function
\begin{equation}
H(s) = C (s I_n -A)^{-1} B +D \nonumber
\end{equation}
is analytic in $\Cset_+$, such that
\begin{equation}
H(i \omega) + H(i \omega)^H \geq 0 \quad \forall \, \omega \in
\Rset \nonumber
\end{equation}
It is well-known \cite{lmi} that the positive-real lemma in linear
matrix inequalty (LMI) format~: $\exists \; P^T=P
> 0$ such that
\begin{equation}
\left[ \begin{array}{cc} A^T P +P A & PB - C^T
\\ B^T P - C & - D -D^T \end{array} \right] \leq 0 \nonumber
\end{equation}
guarantees the passivity of the system (\ref{sub1}). A necessary,
but not sufficient, condition for passivity is that $A$ is stable,
i.e., its eigenvalues are located in the closed left halfplane. In
the sequel we will always suppose that $A$ is Hurwitz stable,
i.e., its eigenvalues are located in the open left halfplane. We
will also assume, unless otherwise stated, that $H(s)$ is
non-passive, and devise ways of finding another as close as
possible passive transfer function $G(s).$ \\
In order to measure how far a given system is from passive we
define the minimum dissipation $\delta_- (H)$ \cite{boyd} as
\begin{equation}
\delta_- (H) = \min_{\omega \in \Rset} \lambda_{\mathrm{min}} \, [
R(\omega) ] \nonumber
\end{equation}
where
\begin{equation}
R(\omega) = H(i \omega) + H(i \omega)^H \nonumber
\end{equation}
Similarly, we also define the maximum dissipation $\delta_+ (H)$
as
\begin{equation}
\delta_+ (H) = \max_{\omega \in \Rset} \lambda_{\mathrm{max}} \, [
R(\omega)] \nonumber
\end{equation}
It is clear that the system is passive if and only if  $\delta_-
(H) \geq 0.$ If $\delta_- (H) < 0$ the system is non-passive, and
if $\delta_+ (H) \leq 0 ,$ the system is anti-passive, in the
sense that then the system with transfer function $-H(s)$ is
passive. \\ In the sequel we will assume, unless otherwise stated,
that the system is non-passive but passifiable, i.e., $-\infty <
\delta_- (H) < 0 < \delta_+ (H) < \infty.$ To obtain $\delta_-
(H)$ (or similarly $\delta_+ (H)$), a simple bisection algorithm,
based on the existence (or non-existence) of imaginary eigenvalues
of the one-parameter Hamiltonian matrix
\begin{equation}
\mathsf{N}_{\delta} = \left[ \begin{array}{cc} A & 0 \\ 0 & -A^T
\end{array} \right] + \left[ \begin{array}{c} B \\ - C^T \end{array} \right] (\delta I_p
- D - D^T)^{-1} \left[ \begin{array}{cc} C & B^T \end{array}
\right] \nonumber
\end{equation}
was proposed in \cite{boyd}. We have
\begin{proposition}
\label{thm0} $\delta > \delta_-(H)$ if and only if $
\mathsf{N}_{\delta}$ admits purely imaginary eigenvalues.
\end{proposition}
\begin{proof} See \cite{boyd}.
\end{proof}
It is clear that Proposition \ref{thm0} always allows to decide,
by checking the eigenvalues of $ \mathsf{N}_{\delta},$ whether
$\delta > \delta_-(H)$ or not. This forms the basis of the
bisection algorithm of \cite{boyd}. The only problem is to start
with a so-called bracket, i.e., provable lower and upper bounds
for
$\delta_- (H) .$ For that purpose we have \\
\begin{proposition}
\label{thm1}
\begin{equation}
-2 \|H\|_{\infty} \leq \delta_- (H) \leq \lambda_{\mathrm{min}} (D
+ D^T) \leq \lambda_{\mathrm{max}} (D + D^T) \leq \delta_+ (H)
\leq 2 \|H\|_{\infty} \label{hinfi}
\end{equation}
\end{proposition}
\begin{proof} Straightforward. Here the infinity norm $
\|H\|_{\infty}$ is defined as
\begin{equation}
\|H\|_{\infty} = \max_{\omega \in \Rset} \|H(i \omega) \|_2
\nonumber
\end{equation}
\end{proof}
Note that we can replace $ \|H\|_{\infty}$ in (\ref{hinfi}) by an
upper bound such as the one given in \cite{boyd}.
\section{MATRIX NEARNESS CONSIDERATIONS}
\begin{thm}
\label{thm2} Let $A=A^H$ be any Hermitian matrix with
eigendecomposition $A=U \Lambda U^H, $ with $U$ a unitary and
$\Lambda$ a real diagonal matrix. Then the positive semidefinite
Hermitian matrix nearest to $A,$ both with respect to the
Frobenius and spectral norms, is given by $A_+ = U \max(0,\Lambda)
U^H .$
\end{thm}
\begin{proof}
First we give the proof for the Frobenius norm. We need to find
\begin{equation}
\min_{X \geq 0} \| X - A \|_F \nonumber
\end{equation}
Putting $X= U Y U^H ,$ and exploiting the unitary invariance of
the Frobenius norm, we obtain
\begin{equation}
\| X - A \|_F^2 = \| Y - \Lambda \|_F^2 = \sum_{i \neq j}
|Y_{ij}|^2 + \sum_i |Y_{ii} - \Lambda_{ii}|^2 \nonumber
\end{equation}
It is clear that the minimum occurs when $Y_{ij}=0$ for $i \neq
j,$ in other words when $Y$ is diagonal. Hence we obtain
\begin{equation}
\| X - A \|_F^2 = \| Y - \Lambda \|_F^2 = \sum_i |Y_{ii} -
\Lambda_{ii}|^2 \nonumber
\end{equation}
It is easy to see that we must take  $Y_{ii} = \max (0,
\Lambda_{ii}) $ and this completes the proof for the Frobenius
norm. Note that
\begin{equation}
\min_{X \geq 0} \| X - A \|_F = \sqrt{ \sum_{\lambda_i (A) <0}
\lambda_i(A)^2} \nonumber
\end{equation}
For the spectral norm, it is known
\cite{Higham89matrixnearness,halmos} that
\begin{equation}
\min_{X \geq 0} \| X - A \|_2= \inf \{ r \geq 0 : A + r I \geq 0
\} \nonumber
\end{equation}
In other words,
\begin{equation}
\min_{X \geq 0} \| X - A \|_2= \max (0, -\lambda_{\mathrm{min}}
(A)) \nonumber
\end{equation}
Now
\begin{equation}
\|A_+ - A \|_2= \max_{\lambda_i (A) <0} |\lambda_i (A)| \nonumber
\end{equation}
which is zero when there are no negative eigenvalues, and
$-\lambda_{\mathrm{min}} (A)$ when there are negative eigenvalues.
\end{proof}
\begin{remark}
\label{rem1} From Theorem \ref{thm2} it is possible to find the
point-wise nearest positive semidefinite matrix for the Hermitian
matrix $R(\omega)= H(i \omega)+ H(i \omega)^H .$ Obviously, if we
decompose $R(\omega)$ as
\begin{equation}
R(\omega) = U(\omega) \Lambda(\omega) U(\omega)^H \nonumber
\end{equation}
then the point-wise nearest positive semidefinite matrix is
\begin{equation}
R_+(\omega) = U(\omega) \max(\Lambda(\omega),0) U(\omega)^H
\nonumber
\end{equation}
Unfortunately, in general, the entries of $R_+(\omega)$ will not
consist of rational functions and therefore cannot represent the
transfer function of an LTI model on the imaginary axis. This
problem, which in fact amounts to a rational approximation
problem, will be addressed in the sequel.
\end{remark}
\section{RATIONAL APPROXIMATIONS}
\begin{thm}
\label{thm3} Let $H(s)$ be passifiable , i.e., $-\infty < \delta_-
(H) < 0 < \delta_+ (H)< \infty,$ and let $R(\omega)= H(i \omega) +
H(i\omega)^H .$ Let further $f(x)$ be a real-rational
function\footnote{A real-rational function $f(x)$ is a rational
function assuming only real values for all real $x.$ } satisfying
\begin{equation}
\alpha \geq f(x) - \max(x,0) \geq 0 \quad \forall x \in
[\delta_-(H),\delta_+(H)] \label{constr1}
\end{equation}
for some finite positive $\alpha .$ Then $f(R(\omega))$ is
positive semidefinite for all $\omega \in \Rset.$ Furthermore we
have
\begin{equation}
\| f(R(\omega)) - R_+ (\omega) \|_2 \leq \alpha \quad \forall
\omega \in \Rset \nonumber
\end{equation}
\end{thm}
\begin{proof}
We have
\begin{equation}
f(R(\omega)) - R_+ (\omega) = U(\omega) \left\{
f(\Lambda(\omega))- \max( \Lambda(\omega),0) \right\} U(\omega)^H
\geq 0 \nonumber
\end{equation}
Since $R_+ (\omega)$ is positive semidefinite, the same holds for
$f(R(\omega)).$ Now, since the spectral norm is unitarily
invariant, we have
\begin{eqnarray}
\| f(R(\omega)) - R_+ (\omega) \|_2 & \leq & \max_i | f(\lambda_i(
\omega)) - \max( \lambda_i (\omega) ,0)| \nonumber \\
 & \leq & \max_{\omega \in \Rset} \max_i | f(\lambda_i(
\omega)) - \max( \lambda_i (\omega) ,0)| \nonumber \\
& \leq & \max _{ x \in [\delta_- (H),\delta_+ (H)]} \{  f(x) -
\max(x,0) \} \leq \alpha \nonumber
\end{eqnarray}
where the last inequality follows from the fact that all
$\lambda_i( \omega)$ are inside the interval $ \left[ \delta_-
(H),\delta_+ (H) \right] .$ This completes the proof.
\end{proof}
Theorem \ref{thm3} shows that the matrix $R_+(\omega)$ can be
approximated from above by the matrix $f(R(\omega)).$ The problem
is to find a suitable real-rational function $f(x).$ We have the
following :
\newline
\begin{thm}
\label{thm5} Let $\zeta_n(x) =x (1+x)^n/((1+x)^n -1).$ Then
\begin{equation}
\frac{1}{n} \geq \zeta_n(x) -\max (x,0) \geq 0 \quad \forall x
\geq -1, \quad n=1,2,3, \dots \nonumber
\end{equation}
\end{thm}
\begin{proof}
First we prove that $\zeta_n(x) -x \geq 0$ for $x \geq 0.$ We have
\begin{equation}
\zeta_n(x) -x = \left( \frac{(1+x)^n-1}{x} \right)^{-1} \nonumber
\end{equation}
which is a positive and decreasing function for $x \geq 0 .$ Next
we prove that $\zeta_n(x)$ is increasing for all $x \geq -1.$ This
is equivalent to proving that $\zeta_n(t-1)=(t^{n+1}
-t^n)/(t^n-1)$ is increasing for all $ t \geq 0.$ This is clearly
the case for $n=1.$ Taking derivatives, we have
\begin{equation}
\frac{ d }{dt} \, \zeta_n(t-1)=\left[n-(n+1)t+t^{n+1} \right] \,
\frac{ t^{n-1}}{(t^n-1)^2} \nonumber
\end{equation}
Now $n-(n+1)t+t^{n+1}$ is $n$ when $t=0$ and $\infty$ when
$t=\infty.$ Since the derivative of $n-(n+1)t+t^{n+1}$ is
$(n+1)(t^n-1),$ the function $n-(n+1)t+t^{n+1}$ attains its unique
minimum (with value zero) at $t=1.$ Hence $\zeta_n(x)$ is
increasing for all $x \geq -1.$ We therefore conclude that
$\zeta_n(x)-\max (x,0)$ increases from 0 to $1/n$ in the interval
$[-1,0],$ and decreases from $1/n$ to 0 in the interval
$[0,\infty],$ which completes the proof.
\end{proof}
\begin{cor} \label{cor1}
Let $H(s)$ be passifiable. Then the real-rational function $f(x) =
\nu \zeta_n( x/ \nu)$ with $\nu = |\delta_- (H)|$ satisfies the
premises of Theorem \ref{thm3} with  $\alpha = \nu/n.$
\end{cor}
\begin{proof}
Straightforward.
\end{proof}
Also, we need to find ways and means to define the matrix $f (
R(\omega))= f ( H(i \omega )+ H(i \omega)^H)$ in the whole
$s-$plane and then to extract a Hurwitz stable transfer function
from it. By analytical continuation, we find the transfer function
$V(s)=f (H(s) + H(-s)^T)$ in the entire $s-$plane. Since $f(x)$ is
real-rational, the transfer function $V(s)$ represents the
realization of a per-symmetric LTI model, \ie satisfying $V(s) =
V(-s)^T .$ This implies that the poles of $V(s)$ admit the
imaginary axis as symmetry axis. The following proposition
indicates how, starting from a per-symmetric LTI model $V(s)$ we
can find a Hurwitz stable transfer function by additive
decomposition \cite{kag,Sajoveli}.
\begin{proposition}
\label{prop1} Let $V(s)$ be per-symmetric, \ie $V(s) = V(-s)^T ,$
such that $V(s)$ has no poles on the imaginary axis. Then $V(s)$
can be decomposed as $V(s) = X(s) + X(-s)^T,$ where $X(s)$ is
Hurwitz stable.
\end{proposition}
\begin{proof}
Putting $V(s) = V_0(s) +D,$ where $V_0(s)$ is strictly proper and
$D=D^T=V(\infty),$ we can decompose $V_0(s)$ uniquely into its
stable and anti-stable parts, i.e.,
\begin{equation}
V_0(s) = X_{stab} (s) + X_{anti} (s) \nonumber
\end{equation}
Since $V_0(s)$ is per-symmetric we have
\begin{equation}
X_{stab} (s) + X_{anti} (s)= X_{stab} (-s)^T + X_{anti} (-s)^T
\nonumber
\end{equation}
and hence $X_{anti} (s)= X_{stab} (-s)^T.$ It follows that $V(s)$
can be decomposed as $V(s) = X(s) + X(-s)^T,$ with $ X(s) =
X_{stab} (s) + \frac{1}{2} D + E ,$ where $E$ is an arbitrary
skew-symmetric matrix. It should be noted that the procedure is
unique when the skew-symmetric matrix $E$ is known a priori.
\end{proof}
\begin{remark}
\label{remv} Proposition \ref{prop1} assumes that $V(s),$ in our
case $V(s)=f (H(s) + H(-s)^T) ,$ does not admit poles on the
imaginary axis. By the inequality constraints (\ref{constr1}) we
know that
\begin{equation}
\alpha \geq f( \lambda_i (\omega)) - \max(\lambda_i(\omega),0)
\geq 0 \label{constr2}
\end{equation}
for all eigenvalues $\lambda_i (\omega)$ of $R(\omega).$ Since
$H(s)$ is assumed Hurwitz stable, $R(\omega)= H(i \omega)+
H(i\omega)^H$ cannot admit real poles, and hence, by the
inequalities (\ref{constr2}), the functions $f( \lambda_i
(\omega))$ are bounded. It follows that all entries of $V(i
\omega) = f(R(\omega))$ are bounded, which implies that $V(s)$
cannot have poles on the imaginary axis.
\end{remark}
In the sequel we will use the Matlab\textregistered\ Robust
Control Toolbox \cite{RCT3} routine \texttt{stabproj} based on the
stable, anti-stable decomposition algorithm \cite{Sajoveli}.
\section{TWO ALGORITHMS}
By Theorem \ref{thm3} and Corollary \ref{cor1} we need to find an
LTI model with transfer function $\phi_n \left( H(s) + H^T (-s)
\right)$ where the real-rational function $\phi_n(x)$ of
denominator degree $n$ and numerator degree $n+1$ is
\begin{equation}
\phi_n(x) = \nu \zeta_n (x/\nu) = \frac{ x (1+x/ \nu)^n}{(1+x/
\nu)^n -1} \nonumber
\end{equation}
where $\nu = |\delta_- (H)| .$ Now it is easy to show that the
following recurrence relationship holds :
\begin{equation}
\phi_{2n}(x) = \frac{ \phi_n(x)^2}{2 \phi_n (x) - x} \quad
n=1,2,3, \ldots \nonumber
\end{equation}
with $\phi_1 (x) = x + \nu .$ \\ A first algorithm ( Algorithm 1)
that comes readily to the mind with $Z(s)=H(s)+H^T(-s)$ is~: \\
Initial value :
\begin{equation}
Z_0(s) = Z(s) + \nu I_p \nonumber
\end{equation}
Loop :
\begin{equation}
\mbox{for}  \quad k=1 \quad \mbox{to} \quad n_1 \, : \quad Z_k (s)
= Z_{k-1} (s) \left( 2 Z_{k-1} (s) - Z(s) \right)^{-1} Z_{k-1} (s)
\nonumber
\end{equation}
It is seen that the associated  $\alpha_k$ upper bound at each
step $Z_k(s), \quad k=0,1,\ldots , n_1 ,$ is $\alpha_k = \nu / 2^k
,$ and all $Z_k(i \omega)$ are, by construction, positive
semidefinite. Since the $Z_k(s)$ are all per-symmetric, we can use
Proposition \ref{prop1} to decompose all (or only the $n_1$th one)
$Z_k(s)$ in their stable and anti-stable parts as
\begin{equation}
Z_k(s) = Z_k^{stab} (s) + Z_k^{stab} (-s)^T \nonumber
\end{equation}
As a last, but necessary step, we must add the skew-symmetric
matrix $\half ( D-D^T) ,$ since this matrix gets deleted when
making the sum $Z(s)=H(s)+H^T(-s).$ In other words, the passive
Hurwitz stable approximant $G_k(s)$ is
\begin{equation}
G_k(s) = Z_k^{stab} (s)  + \half ( D-D^T) \nonumber
\end{equation}
As a very simple illustrative example take $k=0 .$ Since $Z_0(s) =
H(s)+H^T(-s) + \nu I_p ,$ we obtain easily that
\begin{eqnarray}
G_0(s) &=& Z_0^{stab} (s)  + \half ( D-D^T)
\nonumber \\
&=& (H(s)-D) +
 \half (D +D^T + \nu I_p) + \half ( D-D^T) \nonumber \\
 &=& H(s) + \frac{\nu}{2} \, I_p \nonumber
\end{eqnarray}
which is passive by construction. In practice, Algorithm 1 has the
drawback that  the transfer functions $Z_k(s)$ in the algorithmic
loop may not be minimal realizations, and hence it could happen
that the stable anti-stable projection by means of the routine
\texttt{stabproj} might not perform well. \\
Before proposing a second algorithm, and in order to address the
computational complexity of the passivated transfer function
$G(s),$ we want to estimate the number of poles of $G(s) .$ We
suppose that $f(x)$ is an irreducible real-rational function with
denominator degree $M$ and numerator degree $M+1 .$ In this paper
this is always the case, see also Section \ref{oapp}. Hence, if we
suppose that all the poles are simple, we can decompose $f(x)$
into partial fractions as
\begin{equation}
f(x) = \alpha_0 + \beta_0 x + \sum_{k=1}^M \frac{ \alpha_k}{x-
\beta_k} \nonumber
\end{equation}
Now if the original Hurwitz stable transfer function $H(s)$ has
$N$ poles, then the transfer function $Z(s) = H(s)+H(-s)^T$ has
$2N$ poles. Also, $f(Z(s))$ can be written as
\begin{equation}
f(Z(s)) = \alpha_0 I_p + \beta_0 Z(s) + \sum_{k=1}^M \alpha_k (
Z(s)- \beta_k I_p)^{-1} \nonumber
\end{equation}
Hence, the set of poles of $f(Z(s))$ is at most the union of the
sets of poles of $Z(s)$ and $ ( Z(s)- \beta_k I_p)^{-1} .$ It is
well known \cite{CST4}, that when a transfer function $H(s)$ is
such that $H(\infty)$ is invertible, then $H(s)^{-1}$ exists and
has the same number of poles as $H(s).$ Therefore, the number of
poles of $f(Z(s)),$ not considering potential cancellations, is $2
N(M+1).$ Finally, after the stable anti-stable decomposition, this
number is to be divided by two, to yield $N(M+1)$ poles for the
final passivated transfer function $G (s) .$ Of course the number
$N(M+1)$ is only an estimate, since pole-zero cancellations can
occur. If for some reason, the number of poles of the
\textit{explicitly proved passive} transfer function $G (s) $
appears to be unacceptable high, a final judiciously chosen
passivity preserving model order reduction step
\cite{Chen95,Kno05,Sorensen05,Antoulas05,knodfd11} can be applied.
\\ Hence, in order to find a workable algorithm, we have to find
the partial fractions decomposition of $f(x)= \phi_n(x) = \nu
\zeta_n (x/\nu).$ If we restrict ourselves to even $n=2m \geq 2 ,$
we have the partial fraction expansion
\begin{equation}
\zeta_{2m}(x)= x + \frac{1}{m} \, \left( \frac{1}{x+2} + \Re e
\sum_{k=1}^{m-1} \frac{ e^{2 \pi i k/m} - e^{\pi i k/m}}{x+1 -
e^{\pi i k/m}} \right) \nonumber
\end{equation}
Hence, with $f(x) = \nu \zeta_{2m} (x/\nu)$ we have
\begin{equation}
f(Z(s)) = Z(s) + \frac{ \nu^2}{m}  \, \left( [Z(s) + 2 \nu
I_p]^{-1} + \Re e \sum_{k=1}^{m-1} \left( e^{2 \pi i k/m} - e^{\pi
i k/m} \right) [ Z(s)+ \nu ( 1 - e^{\pi i k/m})I_p]^{-1} \right)
\label{zeta12}
\end{equation}
Algorithm 2 performs the state space addition (\ref{zeta12}) as
is, \ie we add the realizations of $Z(s), \; (\nu^2/m) [Z(s) + 2
\nu I_p]^{-1},$ etc., to obtain $f(Z(s)).$ The explicit state
space form for the terms
\begin{equation}
\Re e \left[ \left( e^{2 \pi i k/m} - e^{\pi i k/m} \right) [
Z(s)+ \nu ( 1 - e^{\pi i k/m})I_p]^{-1} \right] \nonumber
\end{equation}
in formula (\ref{zeta12}) is obtained by the state space technique
described in the Appendix. Finally, the stable anti-stable
projection yields the passivated transfer function $G(s).$
\subsection{Numerical Examples}
We will consider only reciprocal non-passive systems, \ie systems
with $H(s)=H(s)^T,$ as these systems are representative of LTI
systems satisfying the electromagnetic condition known as Lorentz
reciprocity \cite{lk2}. Of course the theory also remains valid
for non-reciprocal LTI systems. Since for reciprocal systems
$R(\omega)$ is real and even, this explains why the plots in the
sequel only show values for non-negative frequencies.
\subsubsection{First example}  \noindent As a first
example we take the SISO Hurwitz stable non-passive transfer
function
\begin{equation}
H(s) = \frac{s^5 + 7.2 s^4 + 47.01 s^3 + 230.8 s^2 + 536.6 s +
587.1}{s^5 + 3.2 s^4 + 32.61 s^3 + 43.63 s^2 + 117.5 s + 104.3}
\label{ttp}
\end{equation}
We use the approach of Algorithm 1 with $n_1=2.$ The passivated
approximation $G(s)$ has a non-minimal realization with 65 poles
which are reduced to 20 by the routine \texttt{minreal}
\cite{CST4}. The real and imaginary parts of the original transfer
function $H(s)$ vs. the passivated transfer function $G(s)$ are
shown in Figs \ref{figure1} and \ref{figure2}.
\begin{figure}[htb]
\begin{center}
\includegraphics[height=6cm,width=10cm]{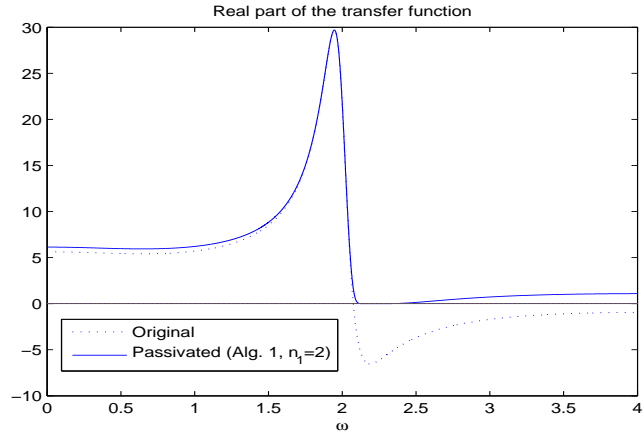}
\caption{Real part of $G(s)$ vs. $H(s)$} \label{figure1}
\end{center}
\end{figure}
\begin{figure}[htb]
\begin{center}
\includegraphics[height=6cm,width=10cm]{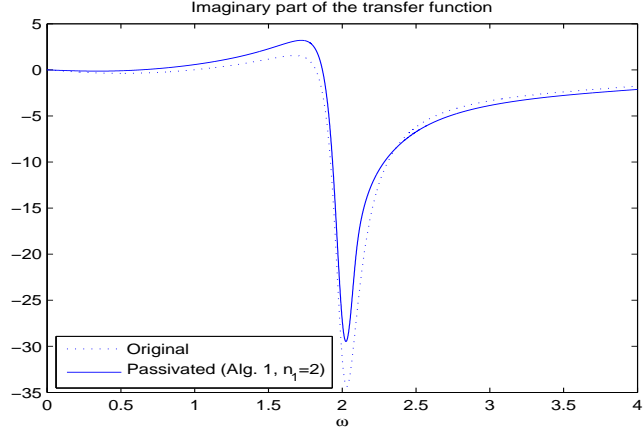}
\caption{Imaginary part of $G(s)$ vs. $H(s)$} \label{figure2}
\end{center}
\end{figure}
\subsubsection{Second example} \noindent As a second example we
take the SISO Hurwitz stable minimum phase non-passive transfer
function
\begin{equation}
H(s)=
\frac{(s+1)(s+3)(s+90)(s+95)(s+100)}{(s+25)(s+35)(s+38)(s+180)(s+185)}
\label{dumi1}
\end{equation}
We use the approach of Algorithm 2 with $m=5.$ The passivated
approximation $G(s)$ has a realization with 50 poles. The real and
imaginary parts of the original transfer function $H(s)$ vs. the
passivated transfer function $G(s)$ are shown in Figs
\ref{figure3} and \ref{figure4}.
\begin{figure}[htb]
\begin{center}
\includegraphics[height=6cm,width=10cm]{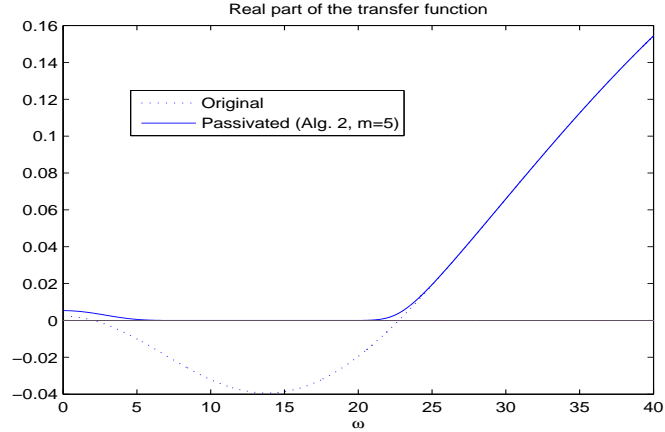}
\caption{Real part of $G(s)$ vs. $H(s)$} \label{figure3}
\end{center}
\end{figure}
\begin{figure}[htb]
\begin{center}
\includegraphics[height=6cm,width=10cm]{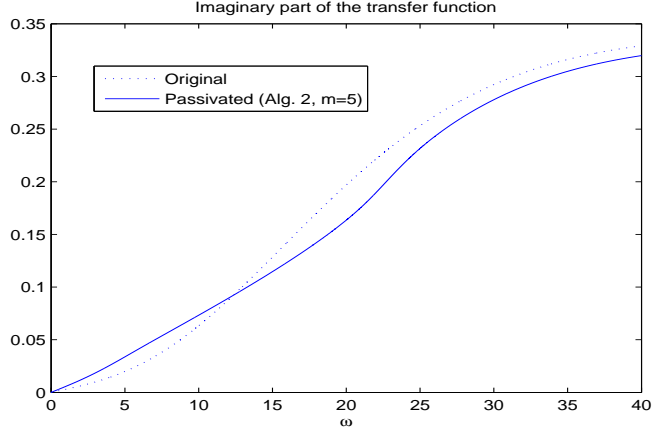}
\caption{Imaginary part of $G(s)$ vs. $H(s)$} \label{figure4}
\end{center}
\end{figure}
\subsubsection{Third example} \noindent As a third example we
take the $2 \times 2$ MIMO Hurwitz stable non-passive transfer
function
\begin{equation}
H(s)= \left[ \begin{array}{cc} 2 + \frac{\DSS 12}{\DSS s^2+3s+2} &
- \frac{\DSS 2s+10}{\DSS s+6} \\ \\ - \frac{ \DSS 2s+10}{ \DSS
s+6} & 2- \frac{\DSS s+3}{\DSS s^2+3s+2}
\end{array} \right] \label{trafe1}
\end{equation}
We use the approach of Algorithm 2 with $m=4.$ The passivated
approximation $G(s)$ has a realization with 48 poles. Fig.
\ref{figure5} plots the values of $\lambda_{min} (G(i \omega)
+G(i\omega)^H)$ vs. $\lambda_{min} (H(i \omega) +H(i\omega)^H).$
To show the nearness of the original and passivated transfer
functions $H(s)$ and $G(s),$ we plot the relative error $\|G(i
\omega) - H(i \omega) \|_2 / \| H(i \omega) \|_2$ in Fig.
\ref{figure6}.
\begin{figure}[htb]
\begin{center}
\includegraphics[height=6cm,width=10cm]{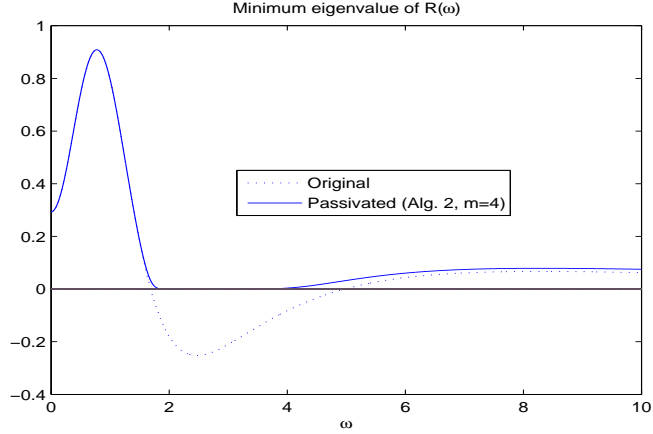}
\caption{Minimum eigenvalues of passivated vs. original transfer
functions} \label{figure5}
\end{center}
\end{figure}
\begin{figure}[htb]
\begin{center}
\includegraphics[height=6cm,width=10cm]{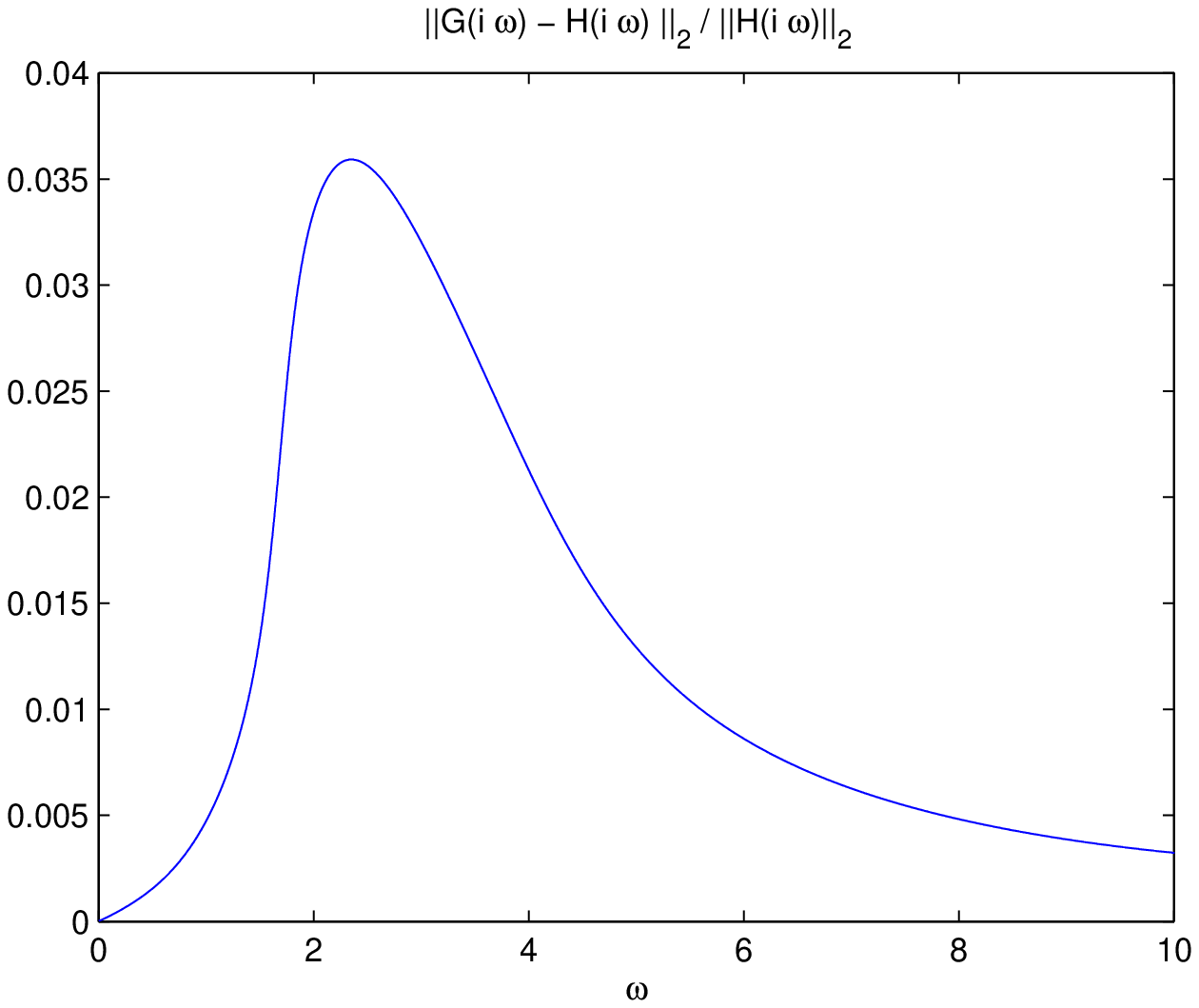}
\caption{Relative error between passivated and original transfer
functions} \label{figure6}
\end{center}
\end{figure}
\section{MINIMAX ALGORITHM} \label{oapp}
The starting point for finding a passive approximant is to find a
real-rational function $f(x)$ that satisfies
\begin{equation}
\alpha \geq f(x) - \max(x,0) \geq 0 \quad \forall x \in [-a,b]
\label{op0}
\end{equation}
where $a= - \delta_-(H) = | \delta_-(H)| $ and $ b= \delta_+(H) .$
Since $\max(x,0)= (|x| + x)/2 ,$ this can be written as
\begin{equation}
\alpha \geq 2 f(x) - x - \alpha -|x| \geq -\alpha \quad \forall x
\in [-a,b] \label{op1}
\end{equation}
Putting $ r(x) = 2 f(x) -x -\alpha ,$ and since our aim is to find
the smallest positive $\alpha$ such that (\ref{op1}) is satisfied,
it is seen that we must find the rational minimax or Chebyshev
approximant, \ie
\begin{equation}
\min_r \max_{x \in [-a,b]} | r(x) - |x| | \nonumber
\end{equation}
Let us first treat the case $a=b=1 ,$ which is well-documented in
the literature \cite{newman1,brutpas1, varga1}. Since $|x|$ is
even and the interval $[-1,1]$ is symmetric with respect to 0, it
is clear that $r(x)$ must be an even rational function, \ie $r(x)
= \rho(x^2) .$ If we take $\rho(t)$ irreducible with numerator and
denominator of exact degree $n ,$ the minimax problem can be
reformulated as:
\begin{equation}
\min_{\rho} \max_{0 \leq t \leq 1} | \rho(t) - \sqrt{t} |
\label{op2}
\end{equation}
Calling $E_n$ the value obtained by the minimax problem
(\ref{op2}), it is clear that at the minimum we must have
\begin{equation}
E_n  \geq \rho(t) - \sqrt{t}  \geq -E_n  \quad \mbox{for} \quad 0
\leq t \leq 1 \label{op3}
\end{equation}
Furthermore, the Remes condition \cite{varga1,ralston1} requires
that there are exactly $2 n +2$ point $t_k$ inside $[0,1]$ where
the equality
\begin{equation}
\sqrt{t_k} -\rho(t_k) = (-1)^k E_n \quad k=1,2, \ldots, 2n+2
\nonumber
\end{equation}
is satisfied. This allows an iterative approach \cite{varga1} to
find the optimal $E_n$ and $\rho(t) .$ The poles and zeros of
$\rho(t)$ are all simple and intertwined on the negative real axis
\cite{blatt1}. It follows that in general $\rho(t)$ can be written
as
\begin{equation}
\rho(t) = a_0 - \sum_{k=1}^n \frac{ a_k}{t+b_k} \nonumber
\end{equation}
where all $a_k, b_k$ are positive. For $n=4$  the coefficients
$a_k, b_k$ with $b_0=E_n$ are given in Table \ref{tab1}.
\begin{table}[htb] \caption{ \label{tab1} Coefficients
for the function $\rho(t)$ for $n=4$ }
\begin{center}
\begin{tabular}{c|c|c|c|c|c}
 $k$ & 0 & 1 & 2 & 3 & 4           \\  \hline
  $a_k$ & 2.6397296257 & 1.4034219887$\times 10^{-6}$
  & 0.0003730797 & 0.0290141901 &
  5.6266532592 \\
  \hline
  $b_k$ & 0.0007365636 & 0.0000917473 & 0.0049831021 &
  0.1014048457 & 2.4866930733
\end{tabular}
\end{center}
\end{table}
\begin{figure}[htb]
\begin{center}
\includegraphics[height=6cm,width=10cm]{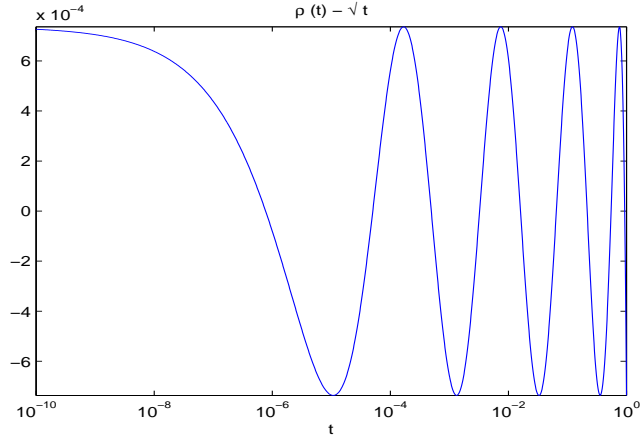}
\caption{Minimax approximation error for $n=4.$} \label{figure7}
\end{center}
\end{figure}
Fig. \ref{figure7} shows the approximation error $\rho(t)-
\sqrt{t}$ and the equioscillation property. Note that the
asymptotic formula of $E_n$ is known \cite{stahl1}, \ie we have $
E_n \approx 8 e^{-\pi \sqrt{2n}}$ for $ n \to \infty .$
\newline Formula (\ref{op3}) implies
\begin{equation}
2 E_n  \geq \rho(x^2)+E_n - |x|  \geq 0  \quad \mbox{ for} \quad
-1 \leq x \leq 1 \nonumber
\end{equation}
or
\begin{equation}
E_n  \geq \frac{\rho(x^2)+ x +E_n}{2} - \max(0,x)  \geq 0  \quad
\mbox{ for} \quad -1 \leq x \leq 1 \label{op4}
\end{equation}
For $a=b=1,$ the best rational function $f(x)$ satisfying
(\ref{op0}) is therefore $f(x) = \half (\rho(x^2)+x+ E_n)$ with
$\alpha=E_n .$ It should be noted that $f(x)$ has numerator degree
$2n +1$ and denominator degree $2 n.$ The case of the general
interval $[-a,b]$ instead of $[-1,1]$ is treated by the following
\begin{thm}
\label{le1} Let $a,b >0$ and $f(x)$ a real-rational function such
that
\begin{equation}
\alpha \geq f(x) - \max(x,0) \geq 0 \quad \mbox{for} \quad -1 \leq
x \leq 1 \nonumber
\end{equation}
Then the real-rational function
\begin{equation}
f_{a,b}(x) = \left[ \frac{x(b-a)+2ab}{b+a} \right] \,f \left(
\frac{x (b+a)}{x(b-a)+2ab} \right) \nonumber
\end{equation}
is such that
\begin{equation}
\alpha \, \max(a,b) \geq f_{a,b}(x)  - \max(x,0) \geq 0 \quad
\mbox{for} \quad -a \leq x \leq b \nonumber
\end{equation}
\end{thm}
\begin{proof}
The bilinear transformation $ g(x) = x (b+a)/ (x(b-a)+2ab)$ maps
the interval $[-a,b]$ onto the interval $[-1,1].$ Moreover, the
linear function $x(b-a)+2ab$ is positive over $[-a,b]$ since it is
positive at the endpoints. Hence
\begin{equation}
\alpha \geq f(g(x)) - \max(g(x),0) \geq 0 \quad \mbox{for} \quad
-a \leq x \leq b \nonumber
\end{equation}
implying
\begin{equation}
\alpha \, \frac{x(b-a)+2 a b}{a+b} \geq \left[
\frac{x(b-a)+2ab}{b+a} \right] \, f(g(x)) - \max(x,0) \geq 0 \quad
\mbox{for} \quad -a \leq x \leq b \nonumber
\end{equation}
which completes the proof. Note that, if the denominator degree of
$f(x)$ is $m$ and the numerator degree is $m+1,$ then the same
holds for $f_{a,b}(x) .$
\end{proof}
In light of formula (\ref{op4}), we take $f(x) = \half
(\rho(x^2)+x+ E_n)$ and $\alpha = E_n .$ The function $f_{a,b}(x)
$ can be conveniently written as
\begin{equation}
f_{a,b}(x) = (\tau x + \kappa) f \left( \frac{x}{\tau x + \kappa}
\right) \nonumber
\end{equation}
where
\begin{equation}
\tau = \frac{\delta_+(H)- | \delta_-(H)|}{\delta_+(H)+ |
\delta_-(H)|}, \qquad \kappa= 2 \frac{\delta_+(H)|
\delta_-(H)|}{\delta_+(H)+ | \delta_-(H)|} \nonumber
\end{equation}
For the function $f^{\ell}(x) = 1/(x^2 + \ell),$ the partial
fraction expansion of $f^{\ell}_{a,b}(x)$ is given by
\begin{equation}
f^{\ell}_{a,b} (x) =\frac{ (\tau x +\kappa)^3}{ x^2 + \ell (\tau x
+ \kappa)^2} = x \frac{ \tau^3}{1+ \tau^2 \ell} + \kappa \, \frac{
\tau^2 (3+ \tau^2 \ell)}{ (1+ \tau^2 \ell)^2} +\Re e \left\{
\frac{\eta(\tau,\kappa,\ell)}{x-\xi(\tau,\kappa,\ell)} \right\}
\nonumber
\end{equation}
where
\begin{equation}
\xi(\tau,\kappa,\ell) = \frac{\kappa \sqrt{\ell}}{i-\tau
\sqrt{\ell}}  \qquad \eta(\tau,\kappa,\ell) = \frac{
\xi(\tau,\kappa,\ell)^3}{\kappa \ell^2} \nonumber
\end{equation}
Hence for the function $f(x) = \half (\rho(x^2)+x+ E_n),$ the
transformed function $f_{a,b}(x)$ can be written as
\begin{equation}
f_{a,b}(x)  = \half \left[ x + (E_n + a_0) (\tau x + \kappa) -
\sum_{k=1}^n a_k f^{b_k}_{a,b} (x)  \right] \label{paf1}
\end{equation}
The partial fraction expansion of (\ref{paf1}) is the key of
Algorithm 3, since we ultimately have to calculate $f_{a,b}(Z(s))
,$ where $Z(s) = H(s)+ H(-s)^T.$ The linear terms of (\ref{paf1})
all add up to the compound linear term
\begin{equation}
f^{linear}_{a,b}(x) =\half \left[ x + (E_n + a_0) (\tau x +
\kappa) - \sum_{k=1}^n a_k \left\{ x \frac{ \tau^3}{1+ \tau^2 b_k}
+ \kappa \, \frac{ \tau^2 (3+ \tau^2 b_k)}{ (1+ \tau^2 b_k)^2}
\right\} \right] \nonumber
\end{equation}
leading to a linear term $ f^{linear}_{a,b}( Z(s)) = k_1 Z(s) +
k_2 I_p.$ The remaining terms, obtained by evaluating
\begin{equation}
\Re e \left\{ \eta(\tau,\kappa, b_k) (Z(s)-\xi(\tau,\kappa,b_k)
I_p )^{-1} \right\} \nonumber
\end{equation}
are obtained by the state space technique described in the
Appendix. Finally, as in Algorithm 2, the stable anti-stable
projection of $f_{a,b}(Z(s))$ is performed in order to obtain the
passivated transfer function $G(s).$
\subsection{Numerical Examples}
\subsubsection{First example} \noindent As our first example we
again take the SISO Hurwitz stable minimum phase non-passive
transfer function (\ref{dumi1}), but here we use Algorithm 3 with
$n=4$ and the coefficients of Table \ref{tab1}. The passivated
approximation $G(s)$ has a realization with 45 poles. The real and
imaginary parts of the original transfer function $H(s)$ vs. the
passivated transfer function $G(s)$ are shown in Figs
\ref{figure8} and \ref{figure9}.
\begin{figure}[htb]
\begin{center}
\includegraphics[height=6cm,width=10cm]{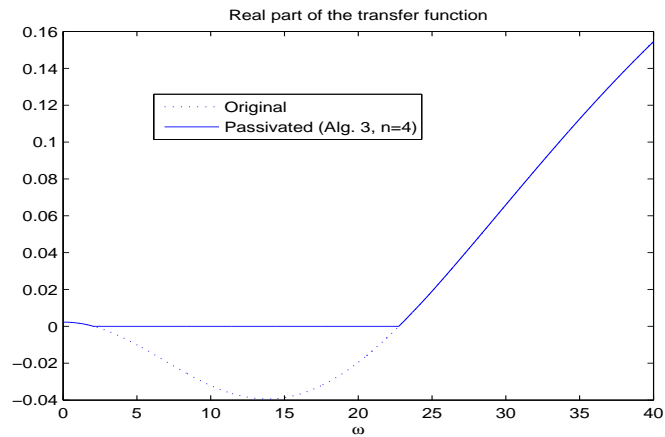}
\caption{Real part of $G(s)$ vs. $H(s)$} \label{figure8}
\end{center}
\end{figure}
\begin{figure}[htb]
\begin{center}
\includegraphics[height=6cm,width=10cm]{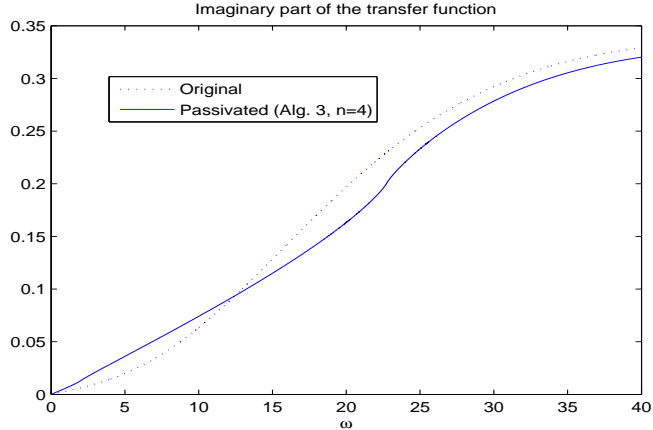}
\caption{Imaginary part of $G(s)$ vs. $H(s)$} \label{figure9}
\end{center}
\end{figure}
It is seen by comparing with Figs \ref{figure3} and \ref{figure4}
that the approximation is better, while requiring 5 poles less.
\subsubsection{Second example} \noindent
\noindent For the second example we again take the MIMO Hurwitz
stable non-passive transfer function (\ref{trafe1}), but here we
use Algorithm 3 with $n=4$ and the coefficients of Table
\ref{tab1}. The passivated approximation $G(s)$ has a realization
with 46 poles. Fig. \ref{figure10} plots the values of
$\lambda_{min} (G(i \omega) +G(i\omega)^H)$ vs. $\lambda_{min}
(H(i \omega) +H(i\omega)^H).$ To show the nearness of the original
and passivated transfer functions $H(s)$ and $G(s),$ we plot the
relative error $\|G(i \omega) - H(i \omega) \|_2 / \| H(i \omega)
\|_2$ in Fig. \ref{figure11}. It is seen by comparing with Figs
\ref{figure5} and \ref{figure6} that the approximation is more or
less similar, but requires 2 poles less.
\begin{figure}[htb]
\begin{center}
\includegraphics[height=6cm,width=10cm]{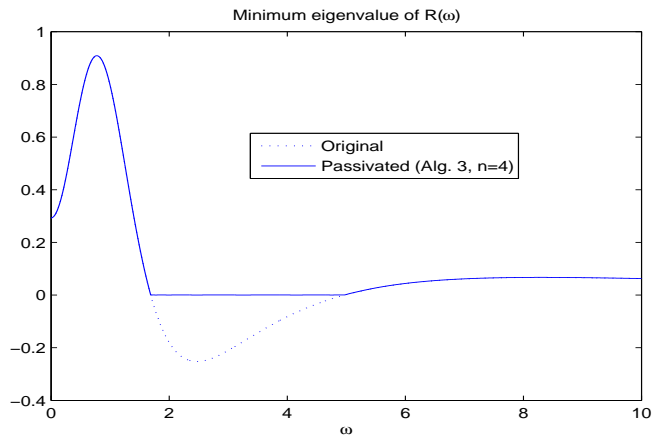}
\caption{Minimum eigenvalues of passivated vs. original transfer
functions} \label{figure10}
\end{center}
\end{figure}
\begin{figure}[htb]
\begin{center}
\includegraphics[height=6cm,width=10cm]{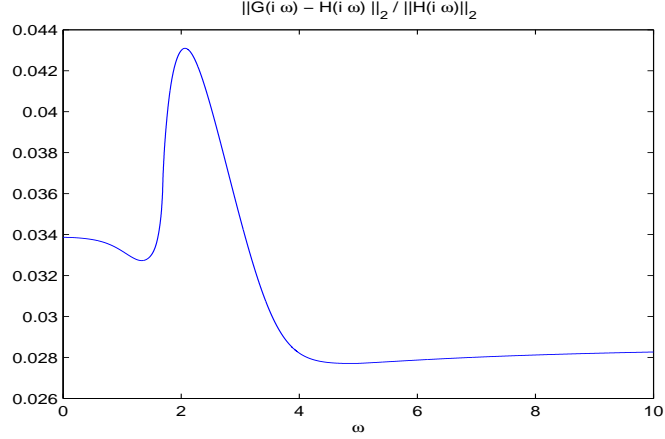}
\caption{Relative error between passivated and original transfer
functions} \label{figure11}
\end{center}
\end{figure}
\section{CONCLUSION}
We have presented a new global passification approach towards
finding a passive transfer function $G(s)$ that is nearest in some
well-defined matrix norm sense to a given non-passive transfer
function $H(s).$ It is shown that the key point in constructing
the nearest passivated transfer function $G(s) ,$ is to find a
good rational approximation to the well-known ramp function over
an interval defined by the minimum and maximum dissipation of the
given non-passive transfer function $H(s).$ It is also shown that
in the Chebyshev or minimax sense this requires finding a rational
Chebyshev approximation of the square root function over the unit
interval. The proposed algorithms rely strongly on the stable
anti-stable projection of a given transfer function. Five
pertinent examples, both SISO and MIMO, are given to show the
accuracy and relevance of the proposed algorithms.
\section{APPENDIX}
Suppose we  have the real-rational transfer function $H(s) =
C(sI_n-A)^{-1}B + D ,$ and we need to evaluate the transfer
function\footnote{Considering $s=d/dt$ as a real operator.}
\begin{equation}
\Re e \, \eta (H(s)- \xi I_p )^{-1} \nonumber
\end{equation}
where $\eta$ and $\xi$ are complex numbers. Suppose also that $ D
- \xi I_p $ is invertible. Putting $ D_{\xi} =  D - \xi I_p ,$ we
have \cite{CST4} that the complex state space transfer function $
\eta (H(s)- \xi I_p )^{-1}$ is given by $ \tilde{C} (s I -
\tilde{A})^{-1} \tilde{B} + \tilde{D}$ where
\begin{equation}
\tilde{A} = A - B D_{\xi}^{-1} C, \quad \tilde{B} = \eta B
D_{\xi}^{-1}, \quad \tilde{C} = -D_{\xi}^{-1} C, \quad \tilde{D} =
\eta D_{\xi}^{-1} \nonumber
\end{equation}
In state space form we have
\begin{eqnarray}
\dot{x} &=& \tilde{A} x + \tilde{B} u  \nonumber \\
y &=& \tilde{C} x + \tilde{D} u \nonumber
\end{eqnarray}
The input $u$ is real, but the output $ y$  is complex. Putting $
y = y_1 + i y_2 ,$ it is clear that we are only interested in
$y_1$ as output. Decomposing all complex vectors and matrices in
their real and imaginary components, we obtain
\begin{eqnarray}
\dot{x}_1  +i \dot{x}_2&=& (\tilde{A}_1 + i \tilde{A}_2) (x_1 + i x_2) + (\tilde{B}_1 + i \tilde{B}_2) u  \nonumber \\
y_1 + i y_2 &=& (\tilde{C}_1 + i \tilde{C}_2) ( x_1 + i x_2) +
(\tilde{D}_1 + i \tilde{D}_2) u \nonumber
\end{eqnarray}
Hence the state space equations with $u$ as input and $y_1$ as
output are simply~:
\begin{eqnarray}
\dot{x}_1 &=& \tilde{A}_1 x_1 - \tilde{A}_2 x_2 + \tilde{B}_1 u  \nonumber \\
\dot{x}_2 &=& \tilde{A}_2 x_1 +  \tilde{A}_1 x_2 +  \tilde{B}_2 u  \nonumber \\
y_1 &=& \tilde{C}_1 x_1 - \tilde{C}_2 x_2 + \tilde{D}_1  u
\nonumber
\end{eqnarray}

\bibliographystyle{elsarticle-num}
\bibliography{Ref_MOR2}

\end{document}